\documentclass[12pt]{article}
\usepackage{setspace}
\usepackage{booktabs}
\usepackage[table]{xcolor}
\usepackage{hyperref}
\usepackage{graphicx}
\usepackage{xcolor}
\definecolor{darkred}{HTML}{950606}
\hypersetup{
    colorlinks=true,
    allcolors=darkred
}
\usepackage{natbib}
\usepackage{amsmath}
\usepackage{amssymb}
\usepackage{amsthm}

\usepackage{algorithm}
\usepackage{algpseudocode}

\title{A General Metric-Space Formulation of the Time Warp Edit Distance (TWED) \\ \Large Technical Note}

\author{Zhen Yi Lau \\lauzyi29@gmail.com}
\date{\today}

\newtheorem{theorem}{Theorem}[section]
\newtheorem{lemma}[theorem]{Lemma}
\newtheorem{proposition}[theorem]{Proposition}
\newtheorem{definition}[theorem]{Definition}
\newtheorem{corollary}[theorem]{Corollary}

\begin{document}

\theoremstyle{remark}
\newtheorem{remark}[theorem]{Remark}

\maketitle

\begin{abstract}
This short technical note presents a formal generalization of the Time Warp Edit Distance (TWED) proposed by Marteau (2009) to arbitrary metric spaces.
By viewing both the observation and temporal domains as metric spaces $(X, d)$ and $(T, \Delta)$, 
we define a Generalized TWED (GTWED) that remains a true metric under mild assumptions. 
We provide self-contained proofs of its metric properties and show that the classical TWED is recovered as a special case when 
$X = \mathbb{R}^d$, $T \subset \mathbb{R}$, and $g(x) = x$. 
This note focuses on the theoretical structure of GTWED and its implications for extending elastic distances beyond time series, which enables the use of TWED-like metrics on sequences over arbitrary domains such as symbolic data, manifolds, or embeddings.
\end{abstract}


\section{Introduction}

The Time Warp Edit Distance (TWED), introduced by \citet{marteau2009twed}, is a metric designed to measure dissimilarity between sequences, particularly time series. Unlike purely elastic measures such as Dynamic Time Warping (DTW) \citep{DTW}, TWED incorporates both temporal elasticity and editing operations (insertions and deletions), controlled by two parameters: a stiffness parameter $\gamma > 0$ and a penalty parameter $\lambda \ge 0$. This formulation ensures that TWED is a true metric—it satisfies non-negativity, symmetry, identity of indiscernibles, and the triangle inequality—when the underlying pointwise distance is itself a metric.

Intuitively, TWED measures the cost of transforming one sequence into another by allowing flexible alignment in time while penalizing both temporal and spatial deviations. Whereas DTW achieves alignment by minimizing cumulative distance through warping paths without any penalty for insertions or deletions, TWED explicitly introduces these penalties to enforce stronger metric behavior and temporal regularization. As a result, TWED can be viewed as a hybrid between elastic alignment and edit-distance-based matching, yielding improved interpretability and robustness to time distortions.

The formulation proposed by \citet{marteau2009twed} defines this distance specifically over discrete time series, where each observation is associated with a scalar timestamp and distances are computed in a Euclidean space. Consequently, TWED is inherently restricted to real-valued, time-indexed data and does not directly generalize to sequences defined over non-numeric or more abstract domains.

To address this limitation, one can generalize TWED beyond the real line by viewing sequences as collections of elements indexed by an ordered metric space $(T, \Delta)$ and valued in another metric space $(X, d)$. In this generalized setting, both $(X, d)$ and $(T, \Delta)$ may represent discrete spaces equipped with custom metrics, or spaces whose elements are themselves structured objects such as vectors or embeddings—allowing comparisons between “vectors of vectors” $(a_1, a_2, \dots)$ where each $a_i \in X$. The metric $\Delta$ provides flexible alignment across indices, enabling the framework to apply to any type of sequential data, not necessarily time series, including symbolic sequences, event streams, or hierarchical representations.

\section{Definition of the General Time Warp Edit Distance (GTWED)}
\label{sec:definition_gtwed}

\paragraph{Metric setting.}
Let $(\mathcal{X}, d)$ be a metric space\footnote{Briefly, a metric space is a set $\mathcal{X}$ equipped with a function $d:\mathcal{X}\times\mathcal{X}\to[0,\infty)$ satisfying: (i) $d(x,y)=0\Leftrightarrow x=y$, (ii) $d(x,y)=d(y,x)$ (symmetry), and (iii) $d(x,z)\le d(x,y)+d(y,z)$ (triangle inequality). See \citep{Tao2023AnalysisII} for details.}, where 
$d : \mathcal{X} \times \mathcal{X} \to [0, \infty)$
is a metric between observations.
Let $(\mathcal{T}, \Delta)$ be a metric space equipped with a total order $\le$,
where $\Delta : \mathcal{T} \times \mathcal{T} \to [0, \infty)$ 
quantifies temporal separation
(e.g.\ $\Delta(t_i, t_j) = |t_i - t_j|$ when $\mathcal{T}\subseteq\mathbb{R}$).

For $p \in \mathbb{N}^{+}$, let 
$t : \{1, \dots, p\} \to \mathcal{T}$ 
be strictly increasing, and let $a_{t(i)} \in \mathcal{X}$ denote
the observation at index $t(i)$.
The set of all finite time-indexed sequences is
\[
\mathcal{U} =
\Bigl\{
A = (a_{t(1)}, a_{t(2)}, \dots, a_{t(p)})
\;\Bigm|\;
p \in \mathbb{N}^{+},\
t(1) < \cdots < t(p)
\Bigr\}
\;\cup\; \{\Omega\},
\]
where $\Omega$ denotes the empty series.

For $0 \le q \le p$, we write
$A^q = (a_{t(1)}, \dots, a_{t(q)})$
with $A^0 = \Omega$.

Let
\[
A^p = (a_{t(1)}, \dots, a_{t(p)}),
\qquad
B^q = (b_{s(1)}, \dots, b_{s(q)}),
\]
with strictly increasing index functions 
$t : \{1,\dots,p\} \to \mathcal{T}$ and 
$s : \{1,\dots,q\} \to \mathcal{T}$.

\paragraph{Metric transform on the local product-space cost.}

Let 
$d : \mathcal{X} \times \mathcal{X} \to [0,\infty)$ 
and 
$\Delta : \mathcal{T} \times \mathcal{T} \to [0,\infty)$ 
be metrics on $\mathcal{X}$ and $\mathcal{T}$, respectively.
For a stiffness parameter $\gamma>0$, define the 
\emph{combined product-space metric}
\[
D(a_{t(i)}, b_{s(j)}) 
  = d(a_{t(i)}, b_{s(j)}) + \gamma\,\Delta(t_i, s_j).
\]

Let $g : [0,\infty) \to [0,\infty)$ 
be an increasing, subadditive function satisfying $g(0)=0$.
Applying $g$ to $D$ yields the 
\emph{regularized local metric}
\[
\tilde D(a_{t(i)}, b_{s(j)}) = g\!\bigl(D(a_{t(i)}, b_{s(j)})\bigr).
\]

Typical choices include
\[
g(x)=x+\alpha\min(x,\tau),
\qquad\text{or}\qquad
g(x)=x+\alpha\bigl(1-e^{-x/\tau}\bigr),
\]
for tunable parameters $\alpha,\tau>0$.

\paragraph{Recursive definition of GTWED.}
For parameters $\lambda\ge0$ (gap penalty) and $\gamma>0$ (stiffness),
define recursively
\[
\delta_{\lambda,\gamma}(A^p,B^q)
=
\min
\begin{aligned}[t]
\bigl\{\, &\delta_{\lambda,\gamma}(A^{p-1},B^q)+\Gamma_A,\\
           &\delta_{\lambda,\gamma}(A^{p-1},B^{q-1})+\Gamma_{A,B},\\
           &\delta_{\lambda,\gamma}(A^{p},B^{q-1})+\Gamma_B
\,\bigr\}.
\end{aligned}
\]

The local operation costs are defined as
\[
\begin{aligned}
\Gamma_A &= \tilde D(a_{t(p)}, a_{t(p-1)}) + \lambda,\\[4pt]
\Gamma_B &= \tilde D(b_{s(q)}, b_{s(q-1)}) + \lambda,\\[4pt]
\Gamma_{A,B} &= 
   \tilde D(a_{t(p)}, b_{s(q)})
 + \tilde D(a_{t(p-1)}, b_{s(q-1)}).
\end{aligned}
\]

The recursion is initialized by
\[
\delta_{\lambda,\gamma}(A^0,B^0)=0,\qquad
\delta_{\lambda,\gamma}(A^p,\Omega)=\sum_{i=1}^{p}\Gamma_A,\qquad
\delta_{\lambda,\gamma}(\Omega,B^q)=\sum_{j=1}^{q}\Gamma_B.
\]

\paragraph{Sentinel convention.}
For notational completeness of the dynamic programming recursion, we extend each sequence
by introducing a \emph{sentinel} element at index~0.
Specifically, we define
\[
a_{t(0)} := a_{t(1)}, \qquad b_{s(0)} := b_{s(1)},
\]
and correspondingly
\[
t(0) := t(1), \qquad s(0) := s(1).
\]
This convention ensures that all local costs
\(
\widetilde{D}(a_{t(i)}, a_{t(i-1)})
\)
and
\(
\widetilde{D}(b_{s(j)}, b_{s(j-1)})
\)
are well-defined even when \(i=1\) or \(j=1\).
Under this setting, the initialization terms
\(\sum \Gamma_A\) and \(\sum \Gamma_B\)
in the recursion evaluate to~0, which makes the dynamic
programming boundary conditions consistent with the metric axioms.

At each step, the dynamic program selects the minimal cumulative cost
to transform prefix $A^p$ into $B^q$ through
(1) deletion of the most recent element of $A$ ($\Gamma_A$),
(2) matching the current elements of $A$ and $B$ ($\Gamma_{A,B}$),
or (3) deletion of the most recent element of $B$ ($\Gamma_B$).

Intuitively, GTWED balances \emph{how different} two samples are in space
(via $d$) and \emph{how far apart in time} they occur (via $\Delta$ and $\gamma$),
while $g$ provides a regularized smoothness control.
Under the above assumptions on $d$, $\Delta$, and $g$, as we will show later, the resulting $\delta_{\lambda,\gamma}$ is a true metric on $\mathcal{U}$.

\section{Metric Properties of GTWED}
\label{sec:metric_proof}

Before proving that the General Time Warp Edit Distance (GTWED) is a metric,
we recall the definition of the classical Time Warp Edit Distance (TWED)
and its known metric properties.

\subsection{Classical TWED and its Metric Property}

\begin{definition}[Traditional TWED, \cite{marteau2009twed}]
Let $(\mathcal{X}, d)$ be a metric space, 
and let $\mathcal{T}\subseteq\mathbb{R}$ be a totally ordered time domain.
For $\lambda\ge0$ (gap penalty) and $\gamma>0$ (stiffness), 
the TWED between two time series 
$A^p=(a_{t(1)},\dots,a_{t(p)})$ and 
$B^q=(b_{s(1)},\dots,b_{s(q)})$ is defined recursively as
\[
\delta_{\lambda,\gamma}^{\mathrm{TWED}}(A^p,B^q)
=
\min
\begin{aligned}[t]
\bigl\{\,
&\delta_{\lambda,\gamma}^{\mathrm{TWED}}(A^{p-1},B^q)
  + d(a_{t(p)},a_{t(p-1)}) + \gamma|t_p-t_{p-1}| + \lambda,\\
&\delta_{\lambda,\gamma}^{\mathrm{TWED}}(A^{p-1},B^{q-1})
  + d(a_{t(p)},b_{s(q)}) + d(a_{t(p-1)},b_{s(q-1)})\\
&\qquad
  + \gamma|t_p-s_q| + \gamma|t_{p-1}-s_{q-1}|,\\
&\delta_{\lambda,\gamma}^{\mathrm{TWED}}(A^{p},B^{q-1})
  + d(b_{s(q)},b_{s(q-1)}) + \gamma|s_q-s_{q-1}| + \lambda
\,\bigr\}.
\end{aligned}
\]

\end{definition}

\begin{proposition}[TWED is a metric]
\label{prop:twed_metric}
If $d$ is a metric and $\gamma>0$, $\lambda\ge0$, 
then $\delta_{\lambda,\gamma}^{\mathrm{TWED}}$ is a metric
on the set of finite time series $\mathcal{U}$.
\end{proposition}

\begin{proof}
See \citet[Proposition~1]{marteau2009twed}.
\end{proof}

\subsection{Supporting Lemmas for GTWED}

\begin{lemma}[Metric Transform Lemma]
\label{lem:metric_transform}
Let $M$ be a metric and 
$g:[0,\infty)\to[0,\infty)$ an increasing, subadditive function
with $g(0)=0$ and $g(x)>0$ for all $x>0$.
Then $\tilde M = g\circ M$ is also a metric.
\end{lemma}

\begin{proof}
For any $x,y,z$,
\[
\begin{aligned}
\tilde M(x,z) &= g(M(x,z)) \\
&\le g(M(x,y)+M(y,z)) \\
&\le g(M(x,y)) + g(M(y,z)) \\
&= \tilde M(x,y) + \tilde M(y,z).
\end{aligned}
\]
using monotonicity and subadditivity of $g$.
Nonnegativity, symmetry, and $g(M(x,y))>0$ for $M(x,y)>0$ follow immediately.
\end{proof}

\begin{lemma}[Positivity of Local Gap Costs]
\label{lem:gap_positive}
For strictly increasing timestamps, $\Delta(t_p,t_{p-1})>0$.
With $\gamma>0$ and $g(x)>0$ for $x>0$,
\[
\Gamma_A = \tilde D(a_{t(p)},a_{t(p-1)})+\lambda>0,
\qquad
\Gamma_B = \tilde D(b_{s(q)},b_{s(q-1)})+\lambda>0.
\]
Hence any insertion or deletion has strictly positive cost.
\end{lemma}

\begin{lemma}[Sum of metrics]
\label{lem:sum_metrics}
Let $(\mathcal{X}, d)$ and $(\mathcal{T}, \Delta)$ be metric spaces,
and let $\alpha,\beta>0$.
Define, for all $x_{t_i}, y_{s_j}\in\mathcal{X}$ with $t_i,s_j\in\mathcal{T}$,
\[
D_{\alpha,\beta}(x_{t_i}, y_{s_j})
= \alpha\, d(x_{t_i}, y_{s_j}) + \beta\, \Delta(t_i, s_j).
\]
Then $D_{\alpha,\beta}$ is a metric on the product domain $\mathcal{X}\times\mathcal{T}$.
\end{lemma}

\begin{proof}
We verify the metric axioms.

\textbf{(i) Nonnegativity and identity of indiscernibles.}
Since $d,\Delta\ge0$ and $\alpha,\beta>0$, we have $D_{\alpha,\beta}\ge0$.
Moreover,
\[
D_{\alpha,\beta}(x_{t_i}, y_{s_j}) = 0
\iff d(x_{t_i}, y_{s_j}) = 0 \text{ and } \Delta(t_i, s_j) = 0
\iff x_{t_i}=y_{s_j}\text{ and } t_i=s_j.
\]

\textbf{(ii) Symmetry.}
By symmetry of $d$ and $\Delta$,
\[
D_{\alpha,\beta}(x_{t_i}, y_{s_j})
= \alpha\,d(x_{t_i}, y_{s_j}) + \beta\,\Delta(t_i, s_j)
= \alpha\,d(y_{s_j}, x_{t_i}) + \beta\,\Delta(s_j, t_i)
= D_{\alpha,\beta}(y_{s_j}, x_{t_i}).
\]

\textbf{(iii) Triangle inequality.}
For any $x_{t_i}, y_{s_j}, z_{u_k}$,
\begin{align*}
D_{\alpha,\beta}(x_{t_i}, z_{u_k})
&= \alpha\,d(x_{t_i}, z_{u_k}) + \beta\,\Delta(t_i, u_k)\\
&\le \alpha\,[d(x_{t_i}, y_{s_j}) + d(y_{s_j}, z_{u_k})]
   + \beta\,[\Delta(t_i, s_j) + \Delta(s_j, u_k)]\\
&= D_{\alpha,\beta}(x_{t_i}, y_{s_j})
 + D_{\alpha,\beta}(y_{s_j}, z_{u_k}),
\end{align*}
using the triangle inequalities of $d$ and $\Delta$.
Hence $D_{\alpha,\beta}$ satisfies all metric axioms.
\end{proof}

\begin{corollary}[Product-space metric for observations and time]
\label{cor:product_metric}
Let $(\mathcal{X}, d)$ and $(\mathcal{T}, \Delta)$ be metric spaces,
and let $\gamma>0$.
Define
\[
D(x_{t_i}, y_{s_j}) = d(x_{t_i}, y_{s_j}) + \gamma\,\Delta(t_i, s_j).
\]
Then $D$ is a metric on $\mathcal{X}\times\mathcal{T}$.
\end{corollary}

\begin{proof}
Immediate from Lemma~\ref{lem:sum_metrics} by setting $\alpha=1$ and $\beta=\gamma$.
\end{proof}

\subsection{GTWED Metric Theorem}

\begin{theorem}[GTWED is a metric]
\label{thm:gtwed_metric}
Let $(\mathcal{X}, d)$ and $(\mathcal{T}, \Delta)$ be metric spaces,
$\gamma>0$, $\lambda\ge0$,
and $g$ satisfy the assumptions of Lemma~\ref{lem:metric_transform}.
Then the recursively defined $\delta_{\lambda,\gamma}$ of 
Section~\ref{sec:definition_gtwed} is a metric on $\mathcal{U}$.
\end{theorem}

\begin{proof}
We prove that $\delta_{\lambda,\gamma}$ satisfies the four metric axioms. 
Throughout, recall $\tilde D := g\circ(d+\gamma\Delta)$ is a metric on $\mathcal{X}\times\mathcal{T}$ by Lemma~\ref{lem:metric_transform}, and that gap costs are strictly positive by Lemma~\ref{lem:gap_positive}. 
Whenever a term like $\Gamma_{A,B}$ (or $\Gamma_A$, $\Gamma_B$, $\Gamma_C$) appears, the involved indices are assumed to be at least $2$ so that the ``previous'' items exist; boundary steps are handled by the initialization and do not require the local inequalities below.

\paragraph{(i) Nonnegativity.}
By construction, the initialization of the recursion is nonnegative,
and all local operation costs are nonnegative 
(Lemma~\ref{lem:gap_positive}).
Consequently, every admissible dynamic-programming path 
has a nonnegative cumulative cost.
Taking the minimum over such paths preserves nonnegativity,
so $\delta_{\lambda,\gamma}\ge0$.

\paragraph{(ii) Identity of indiscernibles.}
$(\Rightarrow)$
If $A=B$ (same values and same timestamps), the path using only matches has zero cost because $\tilde D(x,x)=g(0)=0$. Hence $\delta_{\lambda,\gamma}(A,B)=0$.

$(\Leftarrow)$
Conversely, if $\delta_{\lambda,\gamma}(A,B)=0$, no gap operation can occur since every gap costs $>0$ by Lemma~\ref{lem:gap_positive}. Thus all steps are matches and each match contributes
\[
\Gamma_{A,B}=\tilde D(a_{t(i)},b_{s(i)})+\tilde D(a_{t(i-1)},b_{s(i-1)}).
\]
For the sum to be $0$, every $\tilde D$ term must be $0$, hence $a_{t(i)}=b_{s(i)} $ $\forall i \in \mathbb{N}$. Therefore $A=B$.

\paragraph{(iii) Symmetry.}
The recursion is symmetric in $(A,B)$ and $\tilde D$ is symmetric, so $\delta_{\lambda,\gamma}(A,B)=\delta_{\lambda,\gamma}(B,A)$.

\paragraph{(iv) Triangle inequality.}
We show
\[
\delta_{\lambda,\gamma}(A,C) \;\le\; \delta_{\lambda,\gamma}(A,B) \;+\; \delta_{\lambda,\gamma}(B,C).
\]
The proof is constructive: from an optimal alignment path between $A$ and $B$ and an optimal alignment path between $B$ and $C$, we build a valid alignment path between $A$ and $C$ whose cost is at most the sum of the two optimal costs. Since $\delta_{\lambda,\gamma}(A,C)$ is the minimum over all $A$--$C$ paths, the inequality will follow.

\medskip\noindent
\textit{Notation for paths.}
An alignment path between sequences $X$ and $Y$ is a sequence of grid points 
\(
(0,0)=(i_0,j_0), (i_1,j_1),\dots,(i_m,j_m)=(|X|,|Y|)
\)
with unit moves of three types:
\begin{align*}
\text{Del-}X &: (i,j)\to(i{+}1,j) \quad\text{cost }\Gamma_X,\\
\text{Match} &: (i,j)\to(i{+}1,j{+}1) \quad\text{cost }\Gamma_{X,Y},\\
\text{Del-}Y &: (i,j)\to(i,j{+}1) \quad\text{cost }\Gamma_Y.
\end{align*}
Its cost is the sum of local costs along the path. 
Let $\pi_{AB}$ and $\pi_{BC}$ be \emph{optimal} paths for $\delta_{\lambda,\gamma}(A,B)$ and $\delta_{\lambda,\gamma}(B,C)$, respectively.

\medskip\noindent
\textit{Three local inequalities.}
For any indices where the terms are defined, the metric property of $\tilde D$ implies:
\begin{align}
\Gamma_{A,C} 
&= \tilde D(a_{t(p)},c_{u(r)}) + \tilde D(a_{t(p-1)},c_{u(r-1)}) \nonumber\\
&\le \bigl[\tilde D(a_{t(p)},b_{s(q)}) + \tilde D(b_{s(q)},c_{u(r)})\bigr]
   + \bigl[\tilde D(a_{t(p-1)},b_{s(q-1)}) + \tilde D(b_{s(q-1)},c_{u(r-1)})\bigr] \nonumber\\
&= \Gamma_{A,B} + \Gamma_{B,C},
\label{eq:MM}
\\[4pt]
\Gamma_A 
&= \tilde D(a_{t(p)},a_{t(p-1)}) + \lambda \nonumber\\
&\le \tilde D(a_{t(p)},b_{s(q)}) 
   + \tilde D(b_{s(q)},a_{t(p-1)}) + \lambda \nonumber\\
&\le \tilde D(a_{t(p)},b_{s(q)}) 
   + \tilde D(b_{s(q)},b_{s(q-1)}) 
   + \tilde D(b_{s(q-1)},a_{t(p-1)}) + \lambda \nonumber\\
&= \Gamma_{A,B} + \Gamma_B,
\label{eq:MBtoA}
\\[4pt]
\Gamma_C 
&= \tilde D(c_{u(r)},c_{u(r-1)}) + \lambda \nonumber\\
&\le \tilde D(c_{u(r)},b_{s(q)}) 
   + \tilde D(b_{s(q)},c_{u(r-1)}) + \lambda \nonumber\\
&\le \tilde D(c_{u(r)},b_{s(q)}) 
   + \tilde D(b_{s(q)},b_{s(q-1)}) 
   + \tilde D(b_{s(q-1)},c_{u(r-1)}) + \lambda \nonumber\\
&= \Gamma_B + \Gamma_{B,C}.
\label{eq:BCtoC}
\end{align}
Each line is a direct application of the triangle inequality for the metric $\tilde D$ (possibly twice), plus the fact that $\lambda$ is additive.

\medskip\noindent
\textit{Path composition.}
We now traverse $\pi_{AB}$ and $\pi_{BC}$ \emph{in lockstep on the $B$-index}. 
At any moment we keep a triple of indices $(p,q,r)$ that point to positions in $A,B,C$ reached so far by the partial traversals of $\pi_{AB}$ and $\pi_{BC}$. 
From the current pair of elementary moves in the two paths, we emit an elementary move for an $A$--$C$ path and advance the corresponding indices. There are only six effective pairings (the other three are symmetric or trivial); in each case we bound the emitted $A$--$C$ cost by the \emph{sum} of the two $A$--$B$ and $B$--$C$ costs using \eqref{eq:MM}--\eqref{eq:BCtoC}.

\begin{enumerate}
\item \textbf{(Match, Match):} 
$\pi_{AB}$ uses a match $(p{-}1,q{-}1)\to(p,q)$ and $\pi_{BC}$ uses a match $(q{-}1,r{-}1)\to(q,r)$. 
Emit an $A$--$C$ match $(p{-}1,r{-}1)\to(p,r)$.
By \eqref{eq:MM}, its cost $\Gamma_{A,C}$ is $\le \Gamma_{A,B}+\Gamma_{B,C}$.

\item \textbf{(Match, Del-$B$):} 
$\pi_{AB}$ matches $(p{-}1,q{-}1)\to(p,q)$, while $\pi_{BC}$ deletes $b_{s(q)}$, i.e., $(q{-}1,r)\to(q,r)$ via $\Gamma_B$.
Emit an $A$-deletion $(p{-}1,r)\to(p,r)$. 
By \eqref{eq:MBtoA}, its cost $\Gamma_A \le \Gamma_{A,B}+\Gamma_B$.

\item \textbf{(Del-$B$, Match):} 
$\pi_{AB}$ deletes $b_{s(q)}$, i.e., $(p,q{-}1)\to(p,q)$ via $\Gamma_B$, while $\pi_{BC}$ matches $(q{-}1,r{-}1)\to(q,r)$.
Emit a $C$-deletion $(p,r{-}1)\to(p,r)$. 
By \eqref{eq:BCtoC}, its cost $\Gamma_C \le \Gamma_B+\Gamma_{B,C}$.

\item \textbf{(Del-$A$, Del-$C$):} 
$\pi_{AB}$ does $(p{-}1,q)\to(p,q)$ with cost $\Gamma_A$, and $\pi_{BC}$ does $(q,r{-}1)\to(q,r)$ with cost $\Gamma_C$.
Emit both deletions in either order on the $A$--$C$ path. 
The emitted cost equals $\Gamma_A+\Gamma_C$, which is trivially $\le \Gamma_A+\Gamma_C$.

\item \textbf{(Del-$A$, Match):} 
$\pi_{AB}$ deletes $a_{t(p)}$ and $\pi_{BC}$ matches $(q{-}1,r{-}1)\to(q,r)$. 
Emit the $A$-deletion on the $A$--$C$ path; its cost is $\Gamma_A \le \Gamma_A+\Gamma_{B,C}$.

\item \textbf{(Match, Del-$C$):} 
$\pi_{AB}$ matches $(p{-}1,q{-}1)\to(p,q)$ and $\pi_{BC}$ deletes $c_{u(r)}$.
Emit the $C$-deletion on the $A$--$C$ path; its cost is $\Gamma_C \le \Gamma_{A,B}+\Gamma_C$.
\end{enumerate}

All remaining pairings are symmetric to the above. 
Thus, step by step, the cost emitted for the composed $A$--$C$ path is \emph{no larger} than the sum of the two costs consumed from $\pi_{AB}$ and $\pi_{BC}$ at that step. Summing over the full traversal yields
\[
\mathrm{cost}(\text{composed }A\text{--}C\text{ path}) \;\le\; 
\mathrm{cost}(\pi_{AB}) + \mathrm{cost}(\pi_{BC})
\;=\; \delta_{\lambda,\gamma}(A,B)+\delta_{\lambda,\gamma}(B,C).
\]
Since $\delta_{\lambda,\gamma}(A,C)$ is the minimum over all $A$--$C$ paths,
\[
\delta_{\lambda,\gamma}(A,C)\ \le\ \mathrm{cost}(\text{composed }A\text{--}C\text{ path})
\ \le\ \delta_{\lambda,\gamma}(A,B)+\delta_{\lambda,\gamma}(B,C),
\]
which proves the triangle inequality.

\medskip
Combining (i)–(iv), $\delta_{\lambda,\gamma}$ is a metric on $\mathcal{U}$.
\end{proof}

\begin{remark}[Failure modes]
If $g$ is not strictly positive for $x>0$
(e.g.\ flat near zero),
identity of indiscernibles may fail.
The typical regularizations
$g(x)=x+\alpha\min(x,\tau)$ or 
$g(x)=x+\alpha(1-e^{-x/\tau})$
satisfy $g(x)>0$ for $x>0$, preserving the metric property.
\end{remark}

\begin{proposition}
Let $(X,d)$ be a metric space and $(T,\Delta)$ a metric space for the time domain.  
If we take $X = \mathbb{R}^d$, $T \subseteq \mathbb{R}$, 
$d(x,y) = \lVert x - y \rVert_2$, $\Delta(t(i), t(j)) = |t(i) - t(j)|$, 
and choose $g(x) = x$, 
then the Generalized Time Warp Edit Distance (GTWED) reduces exactly to the Time Warp Edit Distance (TWED) 
as originally defined by \citet{marteau2009twed}.
\end{proposition}

\begin{proof}
Under these parameter choices, the local cost functions of GTWED become
\begin{align*}
\Gamma_A &= d(a_p, a_{p-1}) + \gamma |t_p - t_{p-1}| + \lambda, \\
\Gamma_B &= d(b_q, b_{q-1}) + \gamma |s_q - s_{q-1}| + \lambda, \\
\Gamma_{A,B} &= d(a_p, b_q) + d(a_{p-1}, b_{q-1}) 
+ \gamma \big(|t_p - s_q| + |t_{p-1} - s_{q-1}|\big),
\end{align*}
which are precisely the local costs defining the TWED recursion.
Therefore, the GTWED formulation subsumes TWED as a special case.
\end{proof}

\section{Dynamic-programming Implementation for GTWED}

The General Time Warp Edit Distance (GTWED) can be efficiently computed using dynamic programming. 
The algorithm extends the classical TWED recursion by replacing local costs with the 
regularized metric $\tilde D = g\!\circ\!(d + \gamma \Delta)$. 
A $(p{+}1)\times(q{+}1)$ table stores the minimal cost $\mathrm{DP}[i][j]$ to transform 
prefix $A^i$ into $B^j$. 
Each step selects the minimum among insertion, deletion, and match operations, and 
$\mathrm{DP}[p][q]$ yields $\delta_{\lambda,\gamma}(A,B)$.

\begin{algorithm}[t]
\caption{General Time Warp Edit Distance (GTWED)}
\begin{algorithmic}[1]
\Require 
Sequences $A=(a_{t(1)},\dots,a_{t(p)})$, $B=(b_{s(1)},\dots,b_{s(q)})$
with strictly increasing timestamps; define sentinels $a_{t(0)}{:=}a_{t(1)}$, $b_{s(0)}{:=}b_{s(1)}$; 
increasing subadditive function $g$ with $g(0)=0$.

\Statex
\Function{GTWED}{$A,B,\lambda,\gamma,d,\Delta,g$}

  \State Define $\tilde D(x_{t_i},y_{s_j}) \gets g\!\big(d(x,y) + \gamma\,\Delta(t_i,s_j)\big)$
  \Comment{regularized local metric}

  \State Initialize matrix $\mathrm{DP}[0\!:\!p,\,0\!:\!q]$
  \State $\mathrm{DP}[0][0] \gets 0$

  \Statex
  \Comment{--- Initialization for empty prefixes ---}
  \For{$i=1$ to $p$}
    \State $\mathrm{DP}[i][0] \gets \mathrm{DP}[i-1][0] + \tilde D(a_{t(i)},a_{t(i-1)}) + \lambda$
  \EndFor
  \For{$j=1$ to $q$}
    \State $\mathrm{DP}[0][j] \gets \mathrm{DP}[0][j-1] + \tilde D(b_{s(j)},b_{s(j-1)}) + \lambda$
  \EndFor

  \Statex
  \Comment{--- Main dynamic programming recurrence ---}
  \For{$i=1$ to $p$}
    \For{$j=1$ to $q$}
      \State $c_A \gets \mathrm{DP}[i-1][j] + \tilde D(a_{t(i)},a_{t(i-1)}) + \lambda$
      \Comment{delete $a_{t(i)}$}
      \State $c_M \gets \mathrm{DP}[i-1][j-1] + \tilde D(a_{t(i)},b_{s(j)}) + \tilde D(a_{t(i-1)},b_{s(j-1)})$
      \Comment{match $a_{t(i)}$ with $b_{s(j)}$}
      \State $c_B \gets \mathrm{DP}[i][j-1] + \tilde D(b_{s(j)},b_{s(j-1)}) + \lambda$
      \Comment{delete $b_{s(j)}$}
      \State $\mathrm{DP}[i][j] \gets \min(c_A,\,c_M,\,c_B)$
    \EndFor
  \EndFor

  \State \Return $\mathrm{DP}[p][q]$
\EndFunction

\end{algorithmic}
\end{algorithm}

\clearpage

\paragraph{Implementation notes.}
\begin{itemize}
  \item \textbf{Timestamps:} $t$ and $s$ must be strictly increasing sequences in $\mathcal{T}$.
  \item \textbf{Metrics:} $d$ and $\Delta$ are true metrics; $\tilde D=g\!\circ\!(d+\gamma\Delta)$ is a valid metric on $\mathcal{X}\!\times\!\mathcal{T}$.
  \item \textbf{Initialization:} $\tilde D(a_{t(0)},a_{t(0)})$ and $\tilde D(b_{s(0)},b_{s(0)})$ are set to $0$ by defining $a_{t(0)}\!=a_{t(1)}$, $b_{s(0)}\!=b_{s(1)}$.
  \item \textbf{Complexity:} Time $O(pq)$, space $O(pq)$ (can be reduced to $O(\min(p,q))$ using two rolling rows).
  \item \textbf{Output:} $\delta_{\lambda,\gamma}(A,B) = \mathrm{DP}[p][q]$ satisfies all metric axioms (Theorem~\ref{thm:gtwed_metric}).
\end{itemize}

\begin{remark}[Practical representation]
For implementation convenience, represent each time-indexed observation as a pair $(a_{t(i)},t_i)\in\mathcal{X}\times\mathcal{T}$. This lets you call the local metric directly on tuples, e.g.\ $\tilde D\big((a_{t(i)},t_i),(b_{s(j)},s_j)\big)$, and simplifies DP bookkeeping by keeping values and timestamps together.
\end{remark}

\section{Conclusion}

This report introduced a General Time Warp Edit Distance (GTWED) as a straightforward extension of the Time Warp Edit Distance (TWED) proposed by \citet{marteau2009twed}. While TWED operates on real-valued, time-indexed sequences with scalar timestamps, GTWED generalizes the setting to arbitrary metric spaces for both the observation and index domains. This allows the distance to be defined over discrete spaces with user-defined metrics and to handle structured or multivariate sequence elements in a consistent metric framework.

The main difference between GTWED and TWED lies in the abstraction of the underlying spaces: GTWED treats both temporal and feature domains as general metric spaces rather than real numbers. Apart from this, the recursive dynamic-programming structure and metric properties remain similar.

However, although the theoretical time complexity of GTWED is the same as that of TWED, both requiring $\mathcal{O}(pq)$ operations for sequences of lengths $p$ and $q$, the practical computational cost can be higher. This is because each dynamic programming step in GTWED may involve evaluating user-defined or high-dimensional metrics on structured elements, which can be significantly more expensive than the simple Euclidean computations used in the original TWED.

It is also important to note that the upper bound on the TWED measure proposed by \citet{marteau2009twed}—which relates the distance between original and down-sampled representations—relies crucially on the use of the $L_p$ norm. Since GTWED allows arbitrary metrics for both the feature and index spaces, this assumption no longer holds, and the corresponding upper bound derived for TWED does not directly extend to GTWED. In other words, the geometric guarantees based on $L_p$ properties (such as the additive decomposition and norm-induced convexity) fail to generalize in the metric setting. Therefore, while GTWED retains the same recursive structure and metric validity, the analytical upper bound provided in the original TWED framework is not valid under the generalized formulation.

\bibliography{references.bib}
\end{document}